\documentclass[12pt]{article}
\usepackage{amsfonts}
\usepackage{amssymb, amsthm, amsbsy, amscd, amsfonts, latexsym, euscript,epsf,stackrel}
\usepackage{amscd}
\usepackage{amsmath}
\usepackage{authblk}

\usepackage[utf8]{inputenc}
\usepackage[T2A]{fontenc}

\usepackage[usenames]{color}

\newcommand{\BEQ}{\begin{equation}}
\newcommand{\EEQ}{\end{equation}}
\def\bea{\begin{eqnarray}}
\def\eea{\end{eqnarray}}
\def\nn{\nonumber}

\theoremstyle{plain}
\newtheorem{Th}{Theorem}[section]
\newtheorem{Lem}[Th]{Lemma}

\newtheorem{Hyp}[Th]{Hypothesis}

\newtheorem{Rem}[Th]{Remark}

\theoremstyle{definition}

\def\bea{\begin{eqnarray}}
\def\eea{\end{eqnarray}}

\def\bes{\begin{equation*} \begin{split}}
\def\ees{\end{split} \end{equation*}}

\usepackage{graphicx}

\graphicspath{{pics/}}

\usepackage{hyperref}
\hypersetup{colorlinks=true,linkcolor=blue}

\begin{document}

\title{RE-algebras, quasi-determinants and the full Toda system
}

\author[1,2]{Dmitry V. Talalaev\footnote{dtalalaev@yandex.ru}}


\affil[1]{\small Lomonosov Moscow State University, 
119991, Moscow, Russia.
}
\affil[2]{\small Demidov Yaroslavl State University, 150003, Sovetskaya Str. 14
}

\date{}

\maketitle

\abstract{In 1991, Gelfand and Retakh embodied the idea of a noncommutative Dieudonne determinant in the case of RTT algebra, namely, they found a representation of the quantum determinant of RTT algebra in the form of a product of principal quasi-determinants. In this note we construct an analogue of the above statement for the RE-algebra corresponding to the Drinfeld R-matrix for the order $n=2,3$. Namely, we have found a family of quasi-determinants that are principal with respect to the antidiagonal, commuting among themselves, whose product turns out to be the quantum determinant of this algebra. This family generalizes the construction of integrals of the full Toda system due to Deift et al.  for the quantum case of RE-algebras. In our opinion, this result also clarifies the role of RE-algebras as a quantum homogeneous spaces and can be used to construct effective quantum field theories with a boundary.
}

\section{Introduction}
\subsection{RE-algebras}
The  reflection equation algebras naturally arose in the context of the quantum inverse problem proposed by the Leningrad School. These algebras were first described in the work \cite{Che}, where they appeared in the description of the scattering process for systems with a boundary.
We will consider the so-called Hecke case of the RE-algebra associated with the Drinfeld $R$ matrix
\bea
R=\sum_{i<j}(q-q^{-1})E_{ii}\otimes E_{jj}+\sum_{i,j}q^{\delta(i,j)}E_{ij}\otimes E_{ji}.\nn
\eea
We denote by $\mathcal{A}$  the RE-algebra, defined by the generators $a_{ij}$ and the quadratic relations:
\bea
\label{RE}
RL_1RL_1=L_1RL_1R,
\eea
those both parts are exressions in $Mat_{n\times n}^{\otimes 2}\otimes \mathcal{A},$
the generating matrix $L$ of the generators $a_{ij}$ 
has the form
 \bea
 L=\sum E_{ij}\otimes a_{ij},\nn
 \eea
 the exression $L_1$ means
 \bea
 L_1=\sum E_{ij}\otimes 1 \otimes a_{ij}.\nn
 \eea

 Hereafter $E_{ij}$ are the matrix units in the matrix algebra $Math_{n\times n}.$
This algebra has numerous exceptional properties, it is a deformation of the universal enveloping algebra $U(gl(n))$, contains a center, which is a polynomial algebra with generators deforming classical invariant functions with respect to the $GL(n)$ group action. These and other properties of RE-algebras were studied in \cite{G1}, \cite{G2}, \cite{G3}.

The RE algebra is a comodular algebra over the RTT algebra defined by quadratic relations of the form
\bea
R T_1 T_2=T_1 T_2 R.\nn
\eea
The works \cite{FRT}, \cite{Kol}, \cite{DM} are devoted to this and other manifestations of the structure of the quantum homogeneous space of RE-algebras.

\subsection{Gelfand-Retach quasi-determinants}
The theory of quasi-determinants generalizes the Schur complement technique known in matrix algebra. This theory was constructed in
\cite{GR} and developed in a series of papers dealing with modern problems of cluster manifolds, including a noncommutative version of the Bruhat decomposition \cite{BR1}.
Actually, the quasi-determinant of the matrix $A=(a_{ij})$ of size $n\times n$ with elements in some associative ring is called the expression
\bea
|A|_{pq}=a_{pq}-r_p (A^{pq})^{-1} c_q;\nn
\eea
where $A^{pq}$ is the matrix $A$ with the $p$-th row and $q$-th column deleted, $r_p$ is the $p$-th row with the missing $q$-th element, and $c_q$ is the $q$-th column with the missing $p$-th element. Such a quasi-determinant is defined if the matrix $A^{pq}$ is invertible.

Quasi-determinants have a variety of interesting properties that generalize the properties of minors of  matrices in the commutative case: the so-called heridity property, Plucker, Sylvester identities,  Dodgson condensation principle, etc. These expressions play a role similar to the Plucker coordinates of the quantum analogs of flag varieties.
In \cite{GR}, the $q$-determinant of the RTT algebra was presented in terms of quasi-determinants. Namely, it was shown that
\bea
\label{qdet}
det_q T=|T|_{11} |T^{11}|_{22}\ldots t_{nn}
\eea
and moreover the multipliers in this product commute among themselves.
Let us note that this construction of the noncommutative determinant was proposed by Dieudonne in 1943 \cite{Died}.

\subsection{Full Toda system}
The full symmetric Toda\cite{DLNT} system is a generalization of the open Toda chain. The phase space of the system is the space of symmetric matrices
\bea
L=\left(\begin{array}{cccc}
a_{11} & a_{12} & \ldots & a_{1n} \\
a_{12} & a_{22} & \ldots & a_{2n}\\
\vdots & \vdots & \ddots & \vdots \\
a_{1n} & a_{2n} &\ldots & a_{nn}
\end{array} \right)\nn
\eea
which can be described using the Poisson algebra $S^\cdot(\mathfrak{b}_+)$ with the Kirillov-Kostant bracket. The involutive family for this system is constructed using the so-called chopping procedure. Let's introduce the notation
\bea
\Delta_k(\lambda)=det((L-\lambda Id)_k)\nn
\eea
 where $A_k$ means a submatrix of the matrix $A$ in which the first $k$ columns and the last $k$ rows are removed. One of the main statements of \cite{DLNT} is the theorem on the involutivity of the roots of all polynomials $\Delta_k(\lambda)$. In an alternative form, this statement means that the ratios of coefficients is $\Delta_k(\lambda)$ Poisson commute among themselves and with ratios of coefficients for other $k$. Here we give expressions for the DLNT Hamiltonians in the lowest case $n=3.$ Partial characteristic polynomials take the form:
 \bea
 \Delta_0(\lambda)&=&det(L-\lambda Id)=-\lambda^3+\lambda^2 Tr(L)-\lambda \sigma_2(L) +det(L);\nn\\
 \Delta_1(\lambda)&=&a_{12} a_{23}-(a_{22}-\lambda) a_{13}=\lambda a_{13}+(a_{12} a_{23}-a_{22} a_{13});\nn\\
 \Delta_2(\lambda)&=&a_{13}.\nn
 \eea
The involutive family is generated by:
\bea
I_{0,1}=Tr(L);\quad I_{0,2}=\sigma_2(L),\quad I_{0,3}=det(L),\quad I_{1,1}=(a_{12} a_{23}-a_{22} a_{13})/a_{13}.\nn
\eea
Note that this involutive family is also defined on the entire $S^\cdot(\mathfrak{gl}(n))$. In particular, the work \cite{DLT} is devoted to this.

The full Toda system is, among other things, of interest from the point of view of the geometry of the flag varieties, with its help, in particular, the question of the intersection of real double Bruhat cells was solved in \cite{CST}.

The quantum version of this system for $U(\mathfrak{gl}_n)$ was constructed in \cite{TQT} using quantization of the Gaudin system.

\section{RE-algebra and the commutative family}
\subsection{Commutation relations}
Consider the reflection equation algebra for the Drinfeld $R$-matrix for $n=3$. This algebra is given by the generating matrix
\bea
L=\left(\begin{array}{ccc}
l_1 & a^+ & c^+ \\
a^- & l_2 & b^+  \\
c^- & b^- & l_3
\end{array}\right).
\eea
The equation (\ref{RE}) can be represented in the form of Poincare-Birkhoff-Witt relations, which were kindly provided to the author by Pavel Pyatov \cite{Pyatov} (everywhere below $\lambda=q-q^{-1}$):
\bea
l_i l_j&=&l_j l_i;\nn\\
a^+ l_1&=&q^2 l_1 a^+;\nn\\
a^+ l_2&=&(l_2-\lambda q^{-1}l_1)a^+;\nn\\
a^+ l_3&=&l_3 a^+;\nn\\
a^- l_1&=&q^{-2} l_1 a^-;\nn\\
a^- l_2&=&(l_2+\lambda q^{-3}l_1)a^-;\nn\\
a^- l_3&=&l_3 a^-;\nn
\eea
\bea
b^+ l_1&=&l_1 b^+;\nn\\
b^+ l_2&=&(q^2 l_2-\lambda^2 l_1)b^++\lambda c^+ a^-;\nn\\
b^+ l_3&=&(l_3-\lambda q^{-1}l_2+\lambda^2 q^{-2} l_1)b^+-\lambda q^{-2} c^+ a^-;\nn\\
b^- l_1&=&l_1 b^-;\nn\\
b^- l_2&=&(q^{-2} l_2+\lambda^2q^{-2} l_1)b^--\lambda q^{-2} a^+ c^-;\nn\\
b^- l_3&=&(l_3+\lambda q^{-3}l_2-\lambda^2 q^{-4} l_1)b^-+\lambda q^{-4} a^+ c^-;\nn
\eea
\bea
c^+ l_1&=&q^2 l_1 c^+;\nn\\
c^+ l_2&=&(l_2+\lambda^2 l_1) c^++q\lambda a^+ b^+;\nn\\
c^+ l_3&=&(l_3-\lambda q^{-1}l_1)c^+-\lambda q^{-1} a^+ b^+;\nn\\
c^- l_1&=&q^{-2} l_1 c^-;\nn\\
c^- l_2&=&(l_2-\lambda^2 q^{-2} l_1)c^--q\lambda b^- a^-;\nn\\
c^- l_3&=&(l_3+\lambda q^{-3}l_1)c^-+\lambda q^{-1} b^- a^-;\nn
\eea
\bea
c^+ a^+&=&q a^+ c^+;\nn\\
c^+ b^+&=&q b^+ c^+;\nn\\
b^+ a^+&=&q a^+ b^+ - \lambda l_1 c^+;\nn\\
a^- c^-  &=&q c^- a^-;\nn\\
b^- c^-&=& q c^- b^-;\nn\\
a^- b^- &=& q b^- a^-+\lambda q^{-2} l_1 c^-;\nn
\eea
\bea
a^- b^+ &=&q^{-1}b^+ a^-;\nn\\
a^- c^+&=& q^{-1} c^+ a^- -\lambda q^{-1} l_1 b^+;\nn\\
b^- a^+&=&q^{-1} a^+ b^-;\nn\\
b^- c^+&=&q c^+ b^- -\lambda(l_3-l_2-q^{-2} l_1)a^+;\nn\\
c^- a^+&=&q^{-1} a^+ c^- - \lambda q^{-1} l_1 b^-;\nn\\
c^- b^+&=&q b^+ c^--\lambda(l_3-l_2-q^{-2}l_1)a^-;\nn
\eea
\bea
a^- a^+&=&a^+ a^-+\lambda q^{-1} l_1 (l_1-l_2);\nn\\
c^- c^+&=&c^+ c^- +\lambda a^{-1}a^+ a^-+\lambda q^{-1} l_1(l_1-l_3);\nn\\
b^- b^+&=&b^+ b^-+\lambda q^{-3} a^+ a^- -\lambda q^{-1} c^+ c^-+\lambda q^{-1} (l_2-\lambda q^{-1} l_1)(l_2-l_3).\nn
\eea

\subsection{Properties of quasi-determinants}
Let's introduce the notation $\Delta^{ij}_{kl}$ for the quasi-determinant of the submatrix $L$ defined by rows $i<j$ and columns $k<l$, the highlighted element is always taken in the lower left corner, that is, the one with indexes $j,k$. For example,
\bea
\Delta^{12}_{13}=a^--b^+(c^+)^{-1}l_1.\nn
\eea
Such quasi-determinants are called flag coordinates, they are invariant with respect to the multiplication of the matrix $L$ on the left by a lower-triangular unipotent matrix.
We will also introduce notations for quasi-determinants that are principal in relation to the antidiagonal:
\bea
I_1&=&c^+;\nn\\
I_2&=&\Delta^{12}_{23}=l_2-b^+ (c^+)^{-1} a^+;\nn\\
I_3&=&c^- -(b^-,l_3)\left(\begin{array}{cc}
a^+ & c^+\\
l_2 & b^+
\end{array}\right)^{-1}
\left(
\begin{array}{c}
 l_1 \\
 a^- 
\end{array}
\right).\nn
\eea
By direct calculation, we can show:
\begin{Lem}
\bea
\label{commI1I2}
[I_1,I_2]=0.
\eea
\end{Lem}
The next auxiliary statement is the expression $I_3$ in terms of second-order quasi-determinants, which is a special case of the Dodgson condensation property:
\begin{Lem}
\bea
I_3=\Delta^{23}_{13}-\Delta^{23}_{23}(\Delta^{12}_{23})^{-1}\Delta^{12}_{13}.\nn
\eea
\end{Lem}
The next statement is already specific to RE-algebra:
\begin{Lem}
\label{Lem12131223}
\bea
[\Delta^{12}_{23},\Delta^{12}_{13}]=0.\nn
\eea
\end{Lem}
\begin{proof}
Let's write down the left side of the equality explicitly and use the switching relations to move $(c^+)^{-1}$ out of brackets
\bea
&&(l_2- b^+(c^+)^{-1} a^+)(a^--b^+ (c^+)^{-1} l_1)-(a^--b^+(c^+)^{-1} l_1)(l_2-b^+ (c^+)^{-1}a^+)\nn\\
&=&(c^+)^{-1}\left( (c^+ l_2-q b^+ a^+)(a^- c^+-q^{-2}b^+ l_1)-(c^+ a^--q b^+ l_1)(l_2 c^+-q^{-1} b^+ a^+)\right)(c^+)^{-1}.\nn
\eea
The expression inside the brackets is devoid of fractions
\bea
(c^+ l_2-q b^+ a^+)(a^- c^+-q^{-2}b^+ l_1)-(c^+ a^--q b^+ l_1)(l_2 c^+-q^{-1} b^+ a^+).\nn
\eea
Its triviality is verified by direct calculation.
\end{proof}
Thanks to the lemma \ref{Lem12131223}, the expression for $I_3$ can be represented as:
\bea
I_3=(\Delta^{23}_{13} \Delta^{12}_{23}-\Delta^{23}_{23} \Delta^{12}_{13})I_2^{-1}.\nn
\eea
The above observations lead to the following expression
\bea
Y=I_3 I_2 I_1&=&(\Delta^{23}_{13} \Delta^{12}_{23}-\Delta^{23}_{23} \Delta^{12}_{13})c^+\nn\\
&=&(c^- - l_3 (b^+)^{-1}a^-)(l_2 c^+ - q^{-1} b^+ a^+)-(b^- - l_3 (b^+)^{-1} l_2)(a^- c^+ -q^{-2} b^+ l_1).\nn
\eea
It is easy to see that $Y$ has already been eliminated from fractions.
\begin{Lem}
\label{center}
The element $Y$ lies in the center of $\mathcal{A}$ and is proportional to the quantum determinant of the matrix $L$.
\end{Lem}
\begin{proof}
The centrality of the $Y$ element can be verified by direct calculation. However, it is curious that $Y$ coincides up to a constant with the quantum determinant of the algebra in question \cite{GS1}:
\bea
Y=-q^2 e_3^{(3)}\nn
\eea
where the element $e_3^{(3)}$ can be found by the formulas \cite{Pyatov}
\bea
e_3^{(3)}&=&c^+ b^-a^- +q^{-2} a^+ b^+ c^- +(l_3-\lambda q^{-1}(l_2+  q^{-2} l_1)) e_2^{(2)}\nn\\
&-&l_1 b^+ b^- - q^{-2}(l_2- \lambda(q+q^{-1}) l_1) c^+ c^-;\nn\\
e_2^{(2)}&=&- a^+ a^- +l_1 (l_2-\lambda q^{-1}  l_1).\nn
\eea
\end{proof}

\subsection{The main statement and conclusion}

\begin{Th}
The elements $I_1,I_2,I_3$ commute in the corresponding localization of the algebra $\mathcal{A}.$
\end{Th}
\begin{proof}
It remains to prove that $I_3$ commutes with generators of the first and second order. To do this, we represent $I_3$ as
\bea
I_3=Y (I_1)^{-1} (I_2)^{-1}.\nn
\eea
\end{proof}
First of all, let's formulate a general hypothesis, the proof of which for $n=3$ is the principal result of this note, and which is quite obvious for $n=2.$
\begin{Hyp}
Consider the RE-algebra for the Drinfeld $R$-matrix for an arbitrary $n$
and the quasi-determinants of the form:
\bea
I_{n-k}=|L^{(k)}|_{n-k,k+1},\nn
\eea
where $L^{(k)}$ means the upper right block of the matrix $L$ after cutting out the lower $k$ rows and the left $k$ columns. These elements of the corresponding localization of quantum algebra form a commutative family, in addition, their product
\bea
Y=\prod_{k=1}^n I_k\nn
\eea
lies in the center and is proportional to the quantum determinant.
\end{Hyp}
\begin{Rem}
The connection between RE-algebras and the full Toda system is largely due to the fact that the quadratic Poisson structure of this system, introduced in \cite{ST} and studied in detail in a series of papers \cite{S}, \cite{LP}, \cite{OR}
\bea
\{U\otimes U\}=U_1 U_2 r_1+U_1 r_2 U_2+ U_2 r_3 U_1+r_4 U_1 U_2,\nn
\eea
is the classical limit of the RE-relations \ref{RE} in the neighborhood of unity
\bea
L=1+\varepsilon U.\nn
\eea
\end{Rem}

We assume that this result is directly related to the geometry of quantum homogeneous spaces, an example of which is the RE-algebra. We are also confident that the statement of the general hypothesis about an alternative to the Dieudonne determinant \cite{Died} is directly related to the description of Bruhat cells of this quantum algebra \cite{BR1}.

\section*{Acknowledgements}
The author is grateful to D. Gurevich, V. Rubtsov, V. Sokolov for extremely productive communication that prompted the formulation of the main hypothesis of the work, as well as for the hospitable and fruitful atmosphere of IHES, where most of this research was conducted. The author expresses special gratitude to P. Pyatov for the calculations provided in \cite{Pyatov}. This work was carried out within the framework of a development programme for the Regional Scientific and Educational Mathematical Center of the Yaroslavl State University with financial support from the Ministry of Science and Higher Education of the Russian Federation (Аgreement on provision of subsidy from the federal budget No. 075-02-2024-1442).

\end{document}